\documentclass[11pt]{article}
\usepackage[margin=1in]{geometry}
\usepackage{times}
\usepackage[compact]{titlesec}
\usepackage{amsfonts}
\usepackage{amsmath,amsthm,amssymb}
\usepackage{latexsym}
\usepackage{epic}
\usepackage{epsfig}
\usepackage{hyperref}
\usepackage{verbatim}
\usepackage[justification=centering]{caption}
\usepackage{enumitem}
\usepackage{color}
\usepackage{algorithm}
\usepackage{algorithmic}
\usepackage{sfgame}
\allowdisplaybreaks[1]

\usepackage[numbers]{natbib}

\usepackage{amsfonts}
\usepackage{amsmath,amssymb,mathrsfs}
\usepackage{latexsym}
\usepackage{epic}
\usepackage{epsfig}
\usepackage{hyperref}
\usepackage{verbatim}
\usepackage{enumitem}

\newtheorem{theorem}{Theorem}[section]

\newtheorem{lemma}[theorem]{Lemma}
\newtheorem{proposition}[theorem]{Proposition}

\newtheorem{definition}[theorem]{Definition}

\newtheorem{prop}[theorem]{Proposition}

\newtheorem{assumption}[theorem]{Assumption}
\newtheorem*{conjecture*}{Conjecture}

\newcommand{\cross}{\times}

\newcommand{\set}[1]{\left\{ #1 \right\}}
\newcommand{\union}{\cup}
\newcommand{\intersect}{\cap}

\newcommand{\sm}{\setminus}

\renewcommand{\hat}{\widehat}
\renewcommand{\tilde}{\widetilde}
\renewcommand{\bar}{\overline}

\DeclareMathOperator{\poly}{poly}
\DeclareMathOperator{\polylog}{polylog}

\def\min{\qopname\relax n{min}}
\def\max{\qopname\relax n{max}}

\def\Pr{\qopname\relax n{\mathbf{Pr}}}
\def\Ex{\qopname\relax n{\mathbf{E}}}

\newcommand{\RR}{\mathbb{R}}

\newcommand{\NN}{\mathbb{N}}

\def\A{\mathcal{A}}

\def\G{\mathcal{G}}

\def\P{\mathcal{P}}

\def\T{\mathcal{T}}

\def\eps{\epsilon}
\def\sse{\subseteq}

\newcommand{\eat}[1]{}

\newcommand{\INPUT}{\item[\textbf{Input:}]}
\newcommand{\OUTPUT}{\item[\textbf{Output:}]}

\newcommand{\maxi}[1]{\mbox{maximize} & {#1 } & \\}
\newcommand{\st}{\mbox{subject to} }
\newcommand{\con}[1]{&#1 & \\}
\newcommand{\qcon}[2]{&#1, & \mbox{for } #2.  \\}
\newenvironment{lp}{\begin{equation}  \begin{array}{lll}}{\end{array}\end{equation}}
\newenvironment{lp*}{\begin{equation*}  \begin{array}{lll}}{\end{array}\end{equation*}}

\DeclareMathOperator{\den}{density}

\DeclareMathOperator{\bden}{bi-density}
\newcommand{\sig}{\varphi} 
\newcommand{\di}{\lambda} 
\newcommand{\dii}{x} 

\title{On the Hardness of Signaling}

\author{
Shaddin Dughmi\thanks{Supported in part by NSF CAREER Award CCF-1350900.}\\
Department of Computer Science\\
University of Southern California\\
{\tt shaddin@usc.edu}
}

\begin{document}

\maketitle

\begin{abstract}

There has been a recent surge of interest in the role of information in strategic interactions. Much of this work seeks to understand how the realized equilibrium of a game is influenced by uncertainty in the environment and the information available to players in the game. Lurking beneath this literature is a fundamental, yet largely unexplored, algorithmic question: how should a ``market maker'' who is privy to additional information, and equipped with a specified objective, inform the players in the game? This is an informational analogue of the mechanism design question, and views the \emph{information structure} of a game as a mathematical object to be designed, rather than an exogenous variable.

We initiate a complexity-theoretic examination of the design of optimal information structures in general Bayesian games, a task often referred to as \emph{signaling}. We focus on one of the simplest instantiations of the signaling question: Bayesian zero-sum games, and a principal who must choose an information structure maximizing the equilibrium payoff of one of the players. In this setting, we show that optimal signaling is computationally intractable, and in some cases hard to approximate, assuming that it is hard to recover a planted clique from an Erd\H{o}s-R\'{e}nyi random graph. This is despite the fact that equilibria in these games are computable in polynomial time, and therefore suggests that the hardness of optimal signaling is a distinct phenomenon from the hardness of equilibrium computation.

Necessitated by the non-local nature of information structures, en-route to our results we prove an ``amplification lemma'' for the planted clique problem which may be of independent interest. Specifically, we show that even if we plant many cliques in an Erd\H{o}s-R\'{e}nyi random graph, so much so that most nodes in the graph are in some planted clique, recovering a constant fraction of the planted cliques is no easier than the traditional planted clique problem. 
\end{abstract}





\section{Introduction}

Mechanism design is concerned with designing the rules of a game so as to effect desirable outcomes at equilibrium. This form of intervention, through the design of \emph{incentives}, is fundamentally algorithmic in nature, and has led to a large body of work which examines the computational complexity of incentive-compatible mechanisms. This paper concerns a different, though arguably equally important, mode of intervention: through making available the right \emph{information}.

A classic example illustrating the importance of information in games
is Akerlof's ``market for lemons'' \cite{akerloflemons}.
Each seller in this market is looking to sell a used car which is
equally likely to be a ``peach'' (high quality) or a ``lemon'' (low quality). 
Prospective buyers values peaches at \$1000, and lemons at \$0,
whereas sellers value keeping a peach at \$800 and a lemon at \$0. 
If both buyers and sellers are informed, then peaches are traded at
some price between \$800 and \$1000, 
increasing the social welfare by \$200 per trade. 
In contrast, when only sellers are informed, each buyer would only be
willing to pay his expected value of \$500 for a random car. 
Since this is less than a seller's reservation value for a peach, the
sellers of peaches are driven out of the market, and only lemons are ever traded.

\begin{figure}
\begin{center}
\begin{sfgame}{.6in}{2}{0.9in}
    \players{}{}   
    \col{Coop\ }   \col{Defect}
    \row{Coop\ }           \pay{$-1+\theta$}{$-1 + \theta\ $}   \pay{$-5 +\theta$}{$0$}
    \row{Defect}           \pay{$0$}{$-5+\theta\ $}   \pay{$-4$}{$-4$}
\end{sfgame}
\end{center}
\caption{An incomplete information variant of the prisoners' dilemma.}
\label{fig:prisoners}
\end{figure}

The previous example illustrates that information deficiency can
reduce the payoff of some or all players in a game. 
 Somewhat counter-intuitively, revealing additional information can also degrade the payoffs of some or all players, and in fact optimal information structures may reveal some but not all the information available. For an illustrative example, consider the incomplete-information variant of the classical prisoners' dilemma shown in Figure \ref{fig:prisoners}, in which the game's payoffs are parametrized by a state of nature $\theta$. When $\theta=0$, this is the traditional prisoners' dilemma  in which cooperation is socially optimal, yet the unique Nash equilibrium
is the one where both players defect, making both worse off. When, however, the reward from cooperation is uncertain, the equilibrium depends on players' beliefs about  $\theta$. 

Assuming that $\theta$ is uniformly distributed in $[-3,3]$, and that players are risk-neutral and know nothing about $\theta$ besides its distribution, they play as if $\theta$ equals its expectation  and the defection equilibrium persists. On the other hand, if both players learn the realization of $\theta$ before making their decisions, they would cooperate at equilibrium when $\theta \geq 1$, and defect otherwise, improving both their expected utilities. What is surprising, however, is that neither the opaque nor the transparent signaling scheme is optimal. Consider the following partially-informative scheme: when $\theta > -1$, both players receive the signal \texttt{High}, and otherwise receive the signal \texttt{Low}.  On signal \texttt{High}, the posterior expected value of $\theta$ for both players is $1$, inducing cooperation. On signal \texttt{Low}, the players defect. This information structure induces cooperation with greater probability than either of the opaque or transparent schemes, and thus improves the expected utility of both players.
%
%
%

Motivated by these intricacies, this paper begins a systematic complexity-theoretic examination of optimal information structure design in games, henceforth referred to as \emph{signaling}. 
Our choice of backdrop is the simplest and  most fundamental of all in game theory: 2-player zero-sum games. We consider incomplete-information zero-sum games, in which a state of nature $\theta$ determines the entries of the payoff matrix, and is drawn from a common-knowledge prior distribution $\di$ represented explicitly. We then ask the following question: how should a principal privy to the realization of $\theta$ release information to players in order to effect a desired outcome at minimax equilibrium? Such a principal may be interested in maximizing some weighted combination of the players' utilities, or more generally some (possibly linear) function of their chosen mixed strategies.  Most such natural classes of objective functions include, as a special case, the task of maximizing one player's utility at equilibrium. Therefore, we adopt that as our objective of choice for our hardness results.

Before describing our results and techniques, we elaborate on our modeling assumptions and chosen class of games. First we note that, like much of the prior work, we constrain ourselves to the design of \emph{symmetric} signaling schemes: those which reveal the same information regarding $\theta$ to every player in the game.  We believe such a restriction to a shared communication channel is natural and justified in many foreseeable applications. Moreover, the equilibria of a Bayesian game when agents have symmetric beliefs are simpler to characterize, better understood, and more often admits a canonical choice for equilibrium selection.\footnote{e.g., as apparent in \cite{peaches}, auctions with information asymmetries suffer from the multiplicity of equilibria, and their analysis requires nontrivial refinements of standard equilibrium concepts.} That being said,  the complexity of \emph{asymmetric} signaling schemes is an exciting avenue for future work.\footnote{Asymmetric information introduces interdependencies and correlations in agents' posterior estimates of the game's payoffs. Asymmetric signaling appears related to optimization over the space of correlated equilibria. We refer the reader to the characterization of  \citet{bergemannbce} for more detail.}

Second, we note that in addition to being simple and fundamental, zero-sum games admit  a canonical and \emph{tractable}  choice of equilibrium --- the minimax equilibrium. In addition to simplifying our analysis, this drives a larger conceptual point:  the hardness of optimal signaling is a distinct phenomenon from the hardness of equilibrium computation. Indeed, it is not hard to see that the optimal signaling task must, in general, ``inherit'' the computational complexity of the adopted equilibrium concept. Our choice of a class of games for which equilibrium can be easily computed essentially ``disentangles'' the complexity of optimal signaling from that of equilibrium computation.

\subsection*{Results and Techniques}
Our results are based on the conjectured hardness of the  planted clique problem. Specifically, we assume that it is hard to recover a planted $k$-clique from an Erd\H{o}s-R\'{e}nyi random graph random graph $\G(n,p)$ for $p=\frac{1}{2}$ and some $k=k(n)$ satisfying $k=\omega(\log^2 n)$ and $k=o(\sqrt{n})$.

We prove two main results assuming this conjecture in Section \ref{sec:zero}. For explicitly-represented 2-player Bayesian zero-sum games, we prove that no algorithm computes a player-optimal signaling scheme in polynomial time. We strengthen this to a hardness of approximation result, for an additive absolute constant, for an implicitly-represented zero-sum game which nevertheless permits efficient equilibrium computation. \

Our hardness results hold for a Bayesian zero-sum game in which the states of nature, as well as each player's pure strategies, are the nodes $V$ of a graph $G$. Specifically, we construct such a game in which signaling schemes correspond to ``fractional partitions'' of the states of nature $V$. Roughly speaking, the first player's utility function favors signaling schemes which (fractionally) partition the graph into large yet highly connected clusters of nodes. If $G$ is a random graph in which many $k$-cliques have been planted randomly throughout, so much so that the planted cliques roughly \emph{cover} the nodes of $G$, then the best signaling scheme roughly corresponds to a partition of $G$ into the planted cliques. Moreover, we exhibit an algorithm which, given a near-optimal signaling scheme, recovers a constant fraction of the planted clique cover.

In order to base our results on the planted clique conjecture, in Section \ref{sec:planted} we prove an ``amplification lemma'' extending the hardness of recovering a planted clique to recovering a constant fraction of a planted clique cover. This lemma, which may be of independent interest,  appears surprising because its analogue for the ``distinguishing version'' of the planted cover and planted clique problems is false.


\subsection*{Additional Related Work}

The study of the effects of information on strategic interactions, and mechanisms for signaling, has its roots in the early works of \citet{akerloflemons} and \citet{spence1973job}. 
%
\citet{hirshleifer71} was the first to observe that more information sometimes leads to worse market outcomes, in contrast to earlier work by  \citet{blackwell51} which implied that more information is always better for a single agent in a non-competitive environment. Since then, many works have examined the effects of additional  information on players' equilibrium utilities.  \citet{lehrermediation} showed that additional information improves players' utilities in \emph{common interest} games --- games where players have identical payoffs in each outcome, and  \citet{bassanpositive} exhibited a polyhedral characterization of games in which more information improves the individual utility of \emph{every} player.
%
The work of \citet{pkeski} is related to ours, and considers \emph{asymmetric} information structures in zero-sum games. They show that maximizing one player's utility requires revealing as much information as possible to that player, and withholding as much as possible from his opponent.

Many recent works have established striking characterizations of the space of equilibria attainable by signaling, in natural game domains. \citet{bergemanndiscrimination} considers signaling in a fundamental buyer/seller price-discrimination game, and provides a polyhedral characterization of the space of potential equilibria and their associated payoffs. \citet{CommonValue} characterize the equilibrium of a $2$-player common value auction supplemented with a signaling scheme. 
\citet{bergemannbce} characterize the space of equilibria attainable by (asymmetric) information structures in general games, and relate it to the space of correlated equilibria.


Despite appreciation of the importance of information in strategic interactions, it is only recently that researchers have started viewing the information structure of a game as a mathematical object to be designed, rather than merely an exogenous variable. \citet{bayesianpersuasion} examine settings in which a sender must design a signaling scheme to convince a less informed receiver to take a desired action. Recent work in the CS community, including by \citet{emeksignaling}, \citet{miltersensignals} and \citet{guo}, examines revenue-optimal signaling in an auction setting, and presents polynomial-time algorithms and hardness results for computing it.  \citet{DIR14} examine welfare-optimal signaling in an auction setting under exogenous constraints, and presents polynomial-time algorithms and hardness results.

\section{Preliminaries}
\label{sec:prelims}
\subsection{Games}
\label{sec:games}

At its most general, a \emph{Bayesian game} is a family of \emph{games of complete formation} parametrized by a state of nature $\theta$, where $\theta$ is assumed to be drawn from a known prior distribution. In this paper, we focus on \emph{finite, Bayesian, 2-player zero-sum  games},  each of which is described by the following parameters:
\begin{itemize}
\item Nonnegative integers $r$ and $c$, denoting the number of pure strategies of the row player and column player respectively. 
\item A finite family $\Theta=\set{1,\ldots,M}$ of \emph{states of nature}, which we index by $\theta$.
\item A family of \emph{payoff matrices} $\A^\theta \in \RR^{r \times c}$, indexed by states of nature $\theta \in \Theta$.
\item A \emph{prior distribution} distribution  $\di \in \Delta_{M}$ on the states of nature.
\end{itemize}
Naturally, $\A^\theta(i,j)$ denotes the payoff of the row player when the row player plays $i$, the column player plays $j$, and the state of nature is $\theta$. The column player's payoff in the same situation is $-\A^\theta(i,j)$.

We use two different representations of zero-sum games. In the \emph{explicit representation}, the matrices $\set{A^\theta}_\theta$, as well as the prior distribution $\di$, are given explicitly. We also relax this somewhat for our second result; specifically, we consider a game in which the individual matrices $\A^\theta$ are given \emph{implicitly}, since one of the players has a number of strategies which is exponential in the natural description of the game. Nevertheless, a low-rank bilinear structure permits efficient computation of equilibria, and the value of the game, in the implicitly represented game we consider.\footnote{Implicitly-represented 2-player zero-sum games often admit efficient algorithms for equilibrium computation. Fairly general conditions under which this is possible are well exposited by \citet{dueling}.}



\subsection{Signaling Schemes} \label{prelim:schemes}
We examine policies whereby a principal reveals partial information regarding the state of nature $\theta$ to the players. Crucially, we require that the principal reveal the same information to both players in the game. A \emph{symmetric signaling scheme} is given by a set $\Sigma$ of \emph{signals}, and a (possibly randomized) map $\sig$ from states of nature $\Theta$ to signals $\Sigma$. Abusing notation, we use $\sig(\theta,\sigma)$ to denote the probability of announcing signal $\sigma \in \Sigma$ when the state of nature is $\theta \in \Theta$. We restrict attention to signaling schemes with a finite set of signals $\Sigma$, and this is without loss of generality when $\Theta$ is finite. We elaborate on this after describing the convex decomposition interpretation of a signaling scheme. 

We note that signaling schemes are in one-to-one correspondence with \emph{convex decompositions} of the prior distribution $\di \in \Delta_M$ --- namely, distributions supported on the simplex $\Delta_M$, and having expectation $\di$. 
Formally, a signaling scheme $\sig: \Theta \to \Sigma$ corresponds to the convex decomposition  \[\di = \sum_{\sigma \in \Sigma} \alpha_\sigma \cdot  \dii_\sigma,\] where $\alpha_\sigma= \Pr [\sig(\theta) = \sigma] = \sum_{\theta \in \Theta} \di(\theta) \sig(\theta,\sigma),$ and $\dii_\sigma(\theta) =  \Pr [ \theta | \sig(\theta) = \sigma] = \frac{\di(\theta) \sig(\theta,\sigma)}{\alpha_\sigma}.$
Note that  $\dii_\sigma \in \Delta_M$ is the \emph{posterior distribution} of $\theta$ conditioned on signal $\sigma$, and $\alpha_\sigma$ is the probability of signal $\sigma$. The proof of the converse direction, namely that every convex decomposition of $\di$ corresponds to a signaling scheme, is elementary yet thought provoking, and hence left to the reader.

We judge the quality of a signaling scheme by the outcome it induces signal by signal. Specifically, the principal is equipped with an objective function of the form $\sum_\sigma \alpha_\sigma \cdot f( \dii_\sigma)$, where $f: \Delta_M \to \RR$ is some function mapping a posterior distribution to the quality of the equilibrium chosen by the players. For example, $f$ may be the social welfare at the induced equilibrium, a weighted combination of players' utilities at equilibrium, or something else entirely. In this setup, one can show that there always exists an signaling scheme with a finite set of signals which maximizes our objective, so long as the states of nature are finitely many.  The optimal choice of signaling scheme is related to the \emph{concave envelope}   $f^+$ of the function $f$.\footnote{$f^+$ is the point-wise lowest concave function $h$ for which $h(x) \geq f(x)$ for all $x$ in the domain. Equivalently, the hypograph of $f^+$ is the convex hull of the hypograph of $f$.} 
Specifically, such a signaling scheme achieves $\sum_\sigma \alpha_\sigma \cdot f( \dii_\sigma) = f^+(\di)$. Application of Caratheodory's theorem to the hypograph of $f$, therefore, shows that $M+1$ signals suffice.

Our impossibility results rule out algorithms for computing near optimal signaling schemes, almost agnostic to how such schemes are represented as output. For concreteness, the reader can think of a signaling scheme $\sig$ as represented by the matrix of pairwise probabilities $\sig(\theta,\sigma)$. Since we only consider games where the states of nature, and therefore also the number of signals w.l.o.g., are polynomially many in the description size of the game, this is a compact representation. The representation of $\sig$ as a convex decomposition would do equally well, as both representations can be efficiently computed from each other.  


\subsection{Strategies, Equilibria, and Objectives} \label{prelim:equilibria}


Given a Bayesian game, a symmetric signaling scheme $(\alpha, \dii)$ with signals $\Sigma$ induces $|\Sigma|$ sub-games, one for each signal. The subgame corresponding to signal $\sigma \in \Sigma$ is played with probability $\alpha_\sigma$, and players' (common) beliefs regarding the state of nature in this subgame are given by the posterior distribution $\dii_\sigma \in \Delta_M$. The quality of a symmetric signaling scheme in such a game is contingent on a choice of an \emph{equilibrium concept} and an \emph{objective function}.

\subsubsection*{Equilibrium Concept}

 An equilibrium concept distinguishes a mixed strategy profile for every posterior belief $\dii \in \Delta_M$. This permits defining an objective function on signaling schemes, as described in Section \ref{prelim:schemes}. In general Bayesian games, evaluating the quality of a signaling scheme may be complicated by issues of equilibrium selection --- for example, general sum games often admit many Bayes-Nash equilibria. However, our restriction to two-player zero-sum games side-steps such complications: all standard equilibrium concepts are payoff-equivalent in zero-sum games, and correspond to the well-understood minimax equilibrium in each subgame. 

Our restriction to zero-sum games avoids an additional complication, which arises for more general games. Equilibrium computation in full-information zero-sum games is \emph{tractable}, permitting efficient computation of equilibrium in the Bayesian game, for every posterior distribution over the states of nature.  This avoids ``inheriting'' the computational complexity of equilibrium computation into our optimal signaling problems, and therefore suggests that the hardness of optimal signaling is a distinct phenomenon from the complexity of computing equilibria.


Formally,  
given a Bayesian zero-sum game $(\set{\A^\theta}_{\theta=1}^M, \di)$ as described in Section \ref{sec:games}, and a signaling scheme corresponding to the convex decomposition $(\alpha, \dii)$ of the prior $\di \in \Delta_M$, this naturally induces  a distribution over sub-games of the same form. Specifically, for each signal $\sigma \in \Sigma$, the Bayesian zero-sum game $(\set{\A^\theta}_{\theta=1}^M, \dii_\sigma)$ is played with probability $\alpha_\sigma$. We use $\A^\sigma = \Ex_{\theta \sim \dii_\sigma} [ A^\theta]$ to denote the matrix of posterior expected payoffs conditioned on signal $\sigma$.
%
%
A \emph{Bayesian Nash equilibrium} corresponds to an equilibrium of each sub-game $(\set{\A^\theta}_{\theta=1}^M, \dii_\sigma)$ in which players play as they would in the complete information zero-sum game $\A^\sigma$.  In the sub-game corresponding to signal $\sigma$, the row player's payoff is simply his payoff in the complete information zero-sum game $A^\sigma$, namely
\[ u_r(\A^\sigma) =  \max_{y \in \Delta_r} \min_{j=1}^c (y^\intercal \A^\sigma)_j .\]
The expected payoff of the row player over the entire game is then given by
\[ u_r(\alpha,\dii)= \sum_{\sigma \in \Sigma} \alpha_\sigma u_r(\A^\sigma)\]

\subsubsection*{Objective Function}

We adopt $u_r$, as given above, as our objective function. By symmetry, this is technically equivalent to adopting objective function $u_c$, the utility of the column player.  To justify this choice for our hardness results, observe that most natural classes of objective functions --- say weighted combination of players' utilities, or linear functions of players' mixed strategies --- include $u_r$ and $u_c$ as special cases. 

Additionally,  we are motivated by the fact that the space of all payoff profiles achievable by signaling is spanned by the  two schemes maximizing the utility of one of the players, in the following sense. Let $\sig_r$ be a signaling scheme maximizing the row player's utility, and $\sig_c$ be a signaling scheme maximizing the column player's utility. Moreover, let $u_r^{\max} = -u_c^{\min}$ denote the row player's maximum utility (equivalently, the negation of the column player's minimum utility), and similarly let $u_c^{\max} = -u_r^{\min}$ denote the column player's maximum utility (equivalently, the negation of the row player's minimum utility).  If there exists a signaling scheme $\sig$ inducing a utility profile $(u_r,u_c)$,  then it can be formed as a convex combination of the two extreme signaling schemes $\sig_r$ and $\sig_c$ --- specifically, by running $\sig_r$ with probability $\frac{u_r - u_r^{\min}}{u_r^{\max} - u_r^{\min}}$ and $\sig_c$ otherwise. Consequently,  ``mapping out'' the space of possible payoff profiles, as well constructing a signaling scheme attaining a desired payoff profile, both reduce to computing the extreme schemes $\sig_r$ and~$\sig_c$.

\subsection{Graphs, Clusters, and Density}
An \emph{undirected graph} $G$ is a pair $(V,E)$, where $V$ is a finite set of \emph{nodes} or \emph{vertices}, and $E \sse \binom{V}{2}$ is a set of undirected \emph{edges}. 
We usually use $n$ to denote the number of vertices of a graph, and $m$ to denote the number of edges.
Given a graph $G=(V,E)$, a \emph{cluster} is some $S \sse V$. Given a cluster $S$, we define the intra-cluster edges $E(S)$ as those edges with both endpoints in $S$. The \emph{induced subgraph} of $S$ is the graph $H=(S, E(S))$.  
Moreover, given two clusters $S,T \sse V$, we define the inter-cluster edges $E(S,T)$ as the edges with at least one endpoint in each of $S$ and $T$.



The \emph{density} of $G$ is the fraction of all potential edges in $E$ --- namely $\den(G) = \frac{2|E|}{|V| (|V|-1)}$. More generally, the density of a cluster $S$ in $G$ is the density of the subgraph of $G$ induced by $S$, specifically $\den_G(S) = \frac{2|E(S)|}{|S|(|S| -1)}$. 
   To precisely state and prove our results, we require a slightly different notion of density, defined between pairs of clusters.
We define the \emph{bi-density} between clusters $S$ and $T$ as the fraction of all pairs $(u,v) \in S \cross T$ connected by an edge; formally, $\bden(S,T) = \frac{1}{|S| |T|} \cdot |\set{ (u,v) \in S \cross T | \set{u,v} \in E}|$. Equivalently, the bi-density between $S$ and $T$ can be thought of in terms of the adjacency matrix of $A$ of $G$; specifically,  $\bden(S,T) = \frac{1}{|S| |T|}\sum_{i \in S} \sum_{j \in T} A_{ij}$. Observe that density and bidensity are closely related, in that $\bden(S,S) = \den(S) (1-\frac{1}{|S|})$.



\subsection{Random Graphs}
We make use of \emph{Erd\H{o}s-R\'{e}nyi random graphs}. Given $n \in \NN$ and $p \in [0,1]$, the random graph $\G(n,p)$ has vertices $V=\set{1,\ldots,n}$, and its edges include every $e \in \binom{V}{2}$ independently with probability $p$.  Naturally, for $G \sim \G(n,p)$, every cluster $S$ of $G$ has expected density $p$. Moreover, every pair of clusters $S$ and $T$ has expected bi-density $p \left(1-\frac{|S \intersect T|}{|S| |T|}\right)$.\footnote{The discrepancy from $p$ is due to the absence of self-loops in the drawn graphs.} We make use of the following standard probabilistic bound, proved in Appendix \ref{app:prelim}.



\begin{prop}\label{prop:bidensity-whp}
    Let $p \in (0,1)$ and $\alpha > 1 $ be absolute constants (independent of $n$), and let $G \sim \G(n,p)$. There is an absolute constant $\beta = \beta(p,\alpha)$ such that the following holds with high probability for all clusters $X$ and $Y$ with $|X|,|Y| > \beta \log n$. 
\[ \bden_G(X,Y) \leq \alpha p \]
\end{prop}

As in the above proposition, we make references throughout this paper to guarantees which hold \emph{with high probability} over a family of graphs parametrized by the number of nodes $n$. By this we mean that the claimed property holds with probability at least $1-1/\poly(n)$, for some polynomial in the number of nodes.



\section{Planting Cliques and Clique Covers}
\newcommand{\PC}{{\bf PCLIQUE}}
\newcommand{\PCC}{{\bf PCCOVER}}
\label{sec:planted}

\subsection{The Planted Clique Problem}

For our hardness results, we assume that it is hard to recover a planted clique from an Erd\H{o}s-R\'{e}nyi random graph. Specifically, we consider the following  problem for parameters $n,k \in \NN$ and $p \in [0,1]$.

\begin{definition}[The Planted Clique Problem \PC($n$,$p$,$k$)]
Let $G \sim \G(n,p,k)$ be a random graph with vertices $[n]=\set{1,\ldots,n}$, constructed as follows:
\begin{enumerate}
\item Every edge is included in $G$ independently with probability $p$.
\item A set $S \sse [n]$ with $|S| =k$ is chosen uniformly at random.
\item All edges with both endpoints in $S$ are added to $G$.
\end{enumerate}
Given a sample from $G \sim \G(n,p,k)$, recover the planted clique $S$.
\end{definition}
\noindent
We refer to the parameter $p$ as the \emph{background density}, $S$ as the \emph{planted clique}, and $k$ as the \emph{size} of the clique. Typically, we will think of $k$ as a function of $n$, and $p$ as an absolute constant independent of $n$. Our results will hinge on the conjectured hardness of recovering the clique with constant probability. 
We use the following conjecture as our hardness assumption. 


\begin{assumption}\label{assumption1}
 For some function $k=k(n)$ satisfying $k=\omega(\log^2 n)$ and $k=o(\sqrt{n})$, there is no probabilistic polynomial-time  algorithm for $\PC(n,\frac{1}{2},k)$ with constant success probability.
\end{assumption}
By \emph{constant success probability},  we mean that the probability of the algorithm recovering the clique is bounded below by a constant independent of $n$, over the random draw of the graph $G \sim \G(n,p,k)$ as well as the internal random coins of the algorithm.

The problem of recovering a planted clique, as well as the (no harder) problem of distinguishing a draw from $\G(n,p)$ from a draw from $\G(n,p,k)$, has been the subject of much work since it was introduced by \citet{jerrum92} and \citet{kucera95}. Almost all of this work has focused on $p=\frac{1}{2}$. On the positive side, there is a quasipolynomial time algorithm for recovering the clique when $k \geq 2 \log n$. The best known polynomial-time algorithms, on the other hand, can recover planted cliques of size $\Omega(\sqrt{n})$, through a variety of different algorithmic techniques (e.g. \cite{alonclique,mcsherry,coja,feigekrauthgamer,amesvavasis,feigeron,dekel}).  

Despite this extensive body of work, there are no known polynomial-time algorithms for recovering, or even detecting, planted cliques of size $k=o(\sqrt{n})$. There is evidence the problem is hard: \citet{jerrum92} ruled out  Markov chain approach  for $k= o(\sqrt{n})$; \citet{feigeprobable} ruled out algorithms based on the Lov\'{a}sz-Schrijver family of semi-definite programming relaxations for $k=o(\sqrt{n})$; and recently \citet{feldmanclique} defined a family of ``statistical algorithms,'' and ruled out such algorithms for recovering planted cliques of size $k=o(\sqrt{n})$. Consequently, the planted clique conjecture has been used as a hardness assumption in a variety of different applications (e.g. \cite{alontesting, juels, hazankrauthgamer,  minder}).






\subsection{From Planting Cliques to Planting Covers}

For our results, we construct games in which the random state of nature is a node in some graph $G$. Recalling that a signaling scheme is a kind of fractional partition of the states of nature, with each ``fractional part'' corresponding to a signal, we design our games so that the quality of a signal is proportional to the density of that part of the graph. Since a signaling scheme's quality is measured in aggregate over the entire distribution over states of nature, any hardness result must rule out recovering a family of dense clusters (one per signal) scattered throughout the graph, rather than merely a unique dense cluster as in the planted clique problem.  In this section, we define the \emph{planted clique-cover problem} in which many cliques are planted throughout the graph, and prove a kind ``amplification lemma'' extending the hardness of recovering a planted clique to recovering a constant fraction of the planted cover.

The planted clique-cover problem is parametrized by $n,k,r \in \NN$ and $p \in [0,1]$, and defined as follows.
\begin{definition}[The Planted Clique Cover Problem \PCC($n$,$p$,$k$,$r$)]\label{def:plantedcover}
Let $G \sim \G(n,p,k,r)$ be a random graph on vertices $[n]=\set{1,\ldots,n}$, constructed as follows:
\begin{enumerate}
\item Include every edge in $G$ independently with probability $p$.
\item For $i=1$ to $r$:
  \begin{itemize}
  \item Choose $S_i \sse [n]$ with $|S_i| =k$ uniformly at random.
  \item Add all edges with both endpoints in $S_i$  to $G$.
  \end{itemize}
\end{enumerate}
Given a sample from $G \sim \G(n,p,k,r)$, recover a constant fraction of the planted cliques $S_1,\ldots,S_r$.\footnote{By this we mean that the algorithm should output a list of $k$-cliques in $G$, at least $\alpha r$ of which are in  $\set{S_1,\ldots,S_r}$, for some constant $\alpha$ independent of $n$.}
\end{definition}
As in the planted clique problem, we refer to the parameter $p$ as the \emph{background density}, $S_1,\ldots,S_r$ as the \emph{planted cliques}, $k$ as the \emph{size} of each clique, and $r$ as the \emph{number} of cliques. Moreover, we will think of $k$ and $r$ as functions of $n$, and $p$ as an absolute constant independent of $n$. 
Note that the planted clique problem is the special case of the  planted clique-cover problem when $r=1$. For our results, we will use the planted clique-cover problem for $r = \Theta(\frac{n}{k})$, guaranteeing that a constant fraction of the nodes are in at least one of the planted cliques with high probability. 

At first glance, it might appear that planting many cliques as we do here makes the problem easier than planted clique.  Somewhat surprisingly  --- and we elaborate on why later --- this is not the case. 
We exhibit a reduction from the planted clique problem to the planted clique-cover problem with an arbitrary parameter $r$, showing that indeed the planted clique-cover problem is no easier.  We prove the following lemma.


\begin{lemma}\label{lem:amplification}
  Let $k$ and $r$ be arbitrary functions of $n$, and $p$ be an arbitrary constant. If there is a probabilistic polynomial-time algorithm for $\PCC(n,p,k,r)$ with constant success probability, then there is such an algorithm for $\PC(n,p,k)$.
\end{lemma}
\begin{proof}
The reduction proceeds as follows: Given a graph $G \sim \G(n,p,k)= \G(n,p,k,1)$, we construct a graph $G' \sim \G(n,p,k,r)$ by planting $r-1$ additional $k$-cliques at random, as in Definition \ref{def:plantedcover}. In a sense, we ``continue where planted clique left off'' by adding $r-1$ more cliques placed randomly in the graph. Let $S_1$ denote the original planted clique, and $S_2,\ldots,S_r$ be the ``additional'' cliques planted through the reduction.

The key observation is that the cliques $S_1,\ldots,S_r$ are indistinguishable to any algorithm operating on a sample from $\G(n,p,k,r)$. This is by a symmetry argument: permuting the order in which cliques are planted does not change the distribution $\G(n,p,k,r)$. As a result, any algorithm which recovers a constant fraction of the planted cliques from $G'$ with constant probability must recover each of $S_1,\ldots,S_r$ with constant probability. In particular, applying an algorithm for the planted cover problem to the outcome of our reduction yields a list of cliques which includes $S_1$ --- the original planted clique --- with constant probability. This leads to an algorithm for the planted clique problem with constant success probability.
\end{proof}

Lemma \ref{lem:amplification} is  surprising since it does not appear to hold for the \emph{distinguishing} variants of $\PCC$ and $\PC$. Specifically, when $k=n^\gamma$ for a constant $\gamma < \frac{1}{2}$, whereas it is conjectured that no algorithm can distinguish a sample from $\G(n,\frac{1}{2})$ from a sample from $\G(n,\frac{1}{2},k)$ with constant probability, a simple statistical test succeeds in distinguishing a sample from $\G(n,\frac{1}{2})$ from a sample from $\G(n,\frac{1}{2},k,n/k)$ with constant probability.
The test in question simply counts the edges in the graph. The random graph $\G(n,\frac{1}{2})$ has $\frac{n^2}{4} \pm  O(n)$ edges with constant probability approaching $1$, whereas $\G(n,\frac{1}{2},k,n/k)$ has $\frac{n^2}{4} + \Omega(\frac{n}{k} \cdot k^2) = \frac{n^2}{4} + \Omega(n^{1+\gamma})$ edges in expectation --- well above the confidence interval for the number of edges in $\G(n,\frac{1}{2})$.

\subsection{Approximate Recovery}

To simplify our hardness results, we show that producing a set of vertices with sufficient overlap with a planted clique is polynomial-time equivalent to recovering that entire clique, with high probability. 

\begin{lemma}\label{lem:approx_recovery}
Let $\epsilon > 0$ and $p \leq \frac{1}{2}$ be constants, $k=k(n)$ satisfy $k=\omega(\log^2 n)$ and $k=o(\sqrt{n})$, and $r = O(\frac{n}{k})$. There is a probabilistic polynomial-time algorithm which takes as input $G \sim \G(n,p,k,r)$ and a cluster $T \sse [n]$ , and outputs every planted $k$-clique $S$ in $G$ for which $|T \intersect S| > \epsilon |T \union S|$, with high probability.
\end{lemma}
In other words, to recover a planted $k$-clique $S$ it suffices to produce a set $T$ of size $O(k)$ for which $|T \intersect S| = \Omega(k)$. The analogous statement for the (single) planted clique problem is folklore knowledge. Our proof is similar, though requires some additional accounting because our planted cliques may overlap. We relegate the proof of Lemma \ref{lem:approx_recovery} to  Appendix \ref{app:approx_recovery}.


\section{Zero-Sum Games}
\label{sec:zero}
 
In this section, we show the intractability of optimal signaling in zero sum games, when the objective is maximizing one player's payoff at equilibrium. We prove two results which follow from Assumption \ref{assumption1}: Hardness of optimal signaling  for \emph{explicit} zero sum games --- those games with a polynomial number of strategies and matrices given explicitly, and hardness of approximation for \emph{implicit} zero-sum games. Our former result, for explicit games, requires the real numbers in the payoff matrices to scale with the input size, and therefore does not qualify as a hardness of approximation result per se. Our latter result for implicitly-described games does rule out a constant additive approximation relative to range of possible payoffs, and crucially holds for a game in which equilibrium computation is \emph{tractable}, in the sense described in  Section~\ref{prelim:equilibria}.

\begin{theorem}\label{thm:zerohardness}
  Assumption \ref{assumption1} implies that there is no polynomial-time algorithm which computes a signaling scheme maximizing a player's payoff in an explicitly represented 2-player Bayesian zero-sum game.
\end{theorem}

\begin{theorem}\label{thm:zeroapxhardness}
Assumption \ref{assumption1} implies there is an implicitly-described, yet tractable, Bayesian zero-sum game with payoffs in $[-1,1]$, and an absolute constant $\epsilon$, such that there is no polynomial-time algorithm for computing a signaling scheme which $\eps$-approximately maximizes (in the additive sense) one player's payoff.
\end{theorem}

For both our results, we use variants of the following \emph{security game} involving two players, known as an \emph{attacker} and a \emph{defender}. In a security game, an attacker chooses a target to attack, and a defender chooses targets to defend. The attacker's payoff depends on both his choice of target, and the extent to which said target is protected by the defender. We consider a security game on a network, in which an attacker is looking to take down edges of the network, and a defender is looking to defend those edges. The state of nature determines a vulnerable node $\theta$, and an attacker can only take down edges adjacent to $\theta$ by attacking its other endpoint. The defender, on the other hand, can select a small set of nodes and protect the edges incident on them from attack.

\begin{definition}[Network Security Game]
  Instances of this Bayesian zero-sum security game are described by an undirected graph $G=(V,E)$ with $n$ nodes and $m$ edges, representing a communication network, an integer $d\geq 1$ equal to the number of nodes the defender can simultaneously protect, and a real number $\rho \geq 0$ equal to the utility gain to the defender for protecting a vulnerable or attacked node. States of nature correspond to vertices $V$. When the realized state of nature is $\theta = v$, this indicates that the edges incident on  node $v$ are vulnerable to attack. We assume that $\theta \sim \di$, where the prior $\di$ is the uniform distribution over $V$. The attacker's strategies also correspond to the nodes $V$. The defender's strategies correspond to subsets of $V$ of size at most $d$, representing the choice of nodes to defend. When the state of nature is $\theta \in V$, the attacker attacks $a \in V$, and the defender defends $D \sse V$ with $|D| \leq d$, the payoff of the attacker is defined as follows:
\[   \A^\theta(a,D) = | \{(\theta,a)\} \cap E| - \rho |D \cap \{\theta,a\} |\]


\end{definition}
When $d$ is a constant, the size of the matrices $\A^\theta$ is polynomial in the representation size of the game, and we can think of the game as being represented \emph{explicitly} as a set of $n$ matrices $\A^\theta \in [-2\rho,1]^{n \times \binom{n}{d}}$. However, when $d$ is superconstant, the game matrices have a superpolynomial $\binom{n}{d}$ number of columns. Nevertheless, even then equilibrium computation is still \emph{tractable} in the sense described in Section \ref{prelim:equilibria} --- i.e., our representation permits computation of the equilibrium, and corresponding utilities, for every posterior distribution $\dii \in \Delta_n$ over states of nature. This is a consequence of the ``low dimensional'' nature of the defender's mixed strategy space, and linearity of players' utilities in this low-dimensional representation. Specifically, a mixed strategy of the defender can be summarized as a vector in the matroid  polytope $\P_d = \set{ z \in [0,1]^n : \sum_{i=1}^n z_i \leq d}$. A vector $z \in \P_d$  encodes the probability by which the defender defends each target. It is not hard to see that this representation is loss-less, in that (a)  the vector $z$ summarizing a defender's mixed strategy, together with the attacker's mixed strategy $y \in \Delta_n$, suffices for computing the payoff of each player in a subgame with beliefs $\dii$ --- specifically, using $A$ to denote the adjacency matrix of graph $G$, the attacker's payoff is $x^\intercal A y - \rho (z^\intercal x + z^\intercal y)$;  (b) Given a vector $z \in \P_d$, a corresponding mixed strategy of the defender with small support can be efficiently recovered as the convex decomposition of $z$ into the corner points of the matroid  polytope $\P_d$;\footnote{This is the constructive version of  Caratheodory's theorem, as shown by \citet{GLS}.}  and (c) Given a vector $z \in \P_d$ summarizing the defender's mixed strategy in a subgame with beliefs $\dii$, a best response for the attacker can be computed efficiently.\footnote{Recall from Section \ref{prelim:equilibria} that (a) and (c) are equivalent to the analogous tasks (payoff and equilibrium computation, respectively) in the complete-information zero-sum game $\A^\dii = \Ex_{\theta \sim \dii} \A^\theta$.} Together, observations (a), (b), and (c) imply that the minimax equilibrium of this game, as well as its optimal value, can be computed in time polynomial in $n$ by linear programming.\footnote{For more on  equilibrium computation  in implicitly-described zero-sum games, see  \cite{dueling}.}

Theorems \ref{thm:zerohardness} and \ref{thm:zeroapxhardness} follow from the two lemmas below, as well as Lemma \ref{lem:amplification}.  Specifically, Theorem~\ref{thm:zerohardness} instantiates both lemmas below with $d=1$ and $\rho = \frac{k}{150}$, and Theorem \ref{thm:zeroapxhardness} sets $d=\frac{k}{150}$ and $\rho=1$.

\begin{lemma} \label{lem:zeroopt}
  Let $k$ and $r$ be such that  $r=\frac{3n}{k}$. For a network security game on $G \sim \G(n,\frac{1}{2},k,r)$ with $d,\rho \geq 1$ satisfying  $\rho d =  \frac{k}{150}$, with high probability there is a signaling scheme attaining expected attacker utility at least $0.8$.
\end{lemma}


\begin{lemma} \label{lem:zerorecover}
Let $\epsilon>0$ be an absolute constant, $k$ satisfy $k=\omega(\log^2 n)$ and $k=o(\sqrt{n})$, and $r=\Theta(\frac{n}{k})$. For a network security game on $G \sim \G(n,\frac{1}{2},k,r)$ with with $d,\rho \geq 1$ satisfying  $\rho d = \Theta(k)$, there is a probabilistic polynomial-time algorithm which, given any signaling scheme obtaining attacker utility at least $0.5+ \epsilon$, outputs an  $\Omega(\eps)$ fraction of the planted cliques in $G$, with high probability.
\end{lemma}




\subsection{Proof of Lemma \ref{lem:zeroopt}}
Denote $G = (V,E)$. We construct a signaling scheme which groups together states of nature in the same planted clique, and show that it obtains the claimed attacker utility at equilibrium. Let $S_1,\ldots,S_r$ be the planted $k$-cliques in $G$ listed in some arbitrary order, and let $\hat{S}_i = S_i \sm \union_{1 \leq j < i} S_{j}$. Moreover, let $\hat{S}_0 = V \sm \union_{i=1}^r S_i$ be all remaining vertices not in any of the planted cliques. Evidently, $\hat{S}_0, \hat{S}_1, \ldots, \hat{S}_r$ form a partition of the vertices $V$. Our signaling scheme  $\sig: V \to \set{0,\ldots,r}$ is deterministic, and simply announces $\sig(v) = i$ when $v \in \hat{S}_i$. Represented as a convex decomposition, our signaling scheme is $(\alpha, x)$, where $\alpha_i = \frac{|\hat{S}_i|}{n}$ and $x_i$ is the uniform distribution over $\hat{S}_i$.

To prove that $\sig$ obtains the desired bound with high probability, we need the fact that most nodes are in one of our planted cliques. Recall that each such $S_i$ is a $k$-subset of $V$ chosen independently and uniformly at random. Since $r=3\frac{n}{k}$, in expectation a $1-\frac{1}{e^3}$ fraction of the $n$ vertices are in one of the planted cliques $S_1,\ldots,S_r$. 
This translates to a high-probability guarantee by a standard application of the Chernoff bound and the union bound. Specifically, with high probability we have that $\sum_{i=1}^r \alpha_i = \frac{1}{n} \sum_{i=1}^r \hat{S}_i \geq 0.9$.
%

We can now analyze the utility of an attacker at equilibrium, given the signaling scheme $\sig$ described above. In the subgame corresponding to a signal $i$, the state of nature $\theta$ is uniformly distributed in $\hat{S}_i$. Consider the mixed strategy of the attacker which simply chooses $a \in \hat{S}_i$ uniformly at random.
Since $(\theta,a)$ is a uniformly distributed pair in $\hat{S}_i \cross \hat{S}_i$, the probability that $(\theta,a)$ share an edge is precisely $\bden(\hat{S}_i,\hat{S}_i)$. Moreover, for any defender strategy $D \sse V$ with $|D|\leq d$, the sum of $\Pr [ a \in D]$ and $\Pr [ \theta \in D]$ is at most $2d/|\hat{S}_i|$. Therefore,  any pure defender strategy $D$ begets attacker utility $u^i (\sig) \geq \bden(\hat{S}_i,\hat{S}_i)-\frac{2d \rho}{|\hat{S}_i|}$. For $i\neq 0$, since $\hat{S}_i$ forms a clique this is at least $1-\frac{2d \rho +1}{|\hat{S}_i|} $. 

We can now complete the proof by bounding the expected utility of the attacker from signaling scheme $\sig$ by his utility when he simply plays the uniform strategy over $\hat{S}_i$ for each subgame $i$.
\begin{align*}
  u(\sig) &= \sum_{i=0}^r \alpha_i u^i(\sig) \\
&\geq \sum_{i=0}^r  \alpha_i  \left( \bden(\hat{S}_i,\hat{S}_i)-\frac{2d \rho}{|\hat{S}_i|}\right)\\
&=  \left( \sum_{i=0}^r  \alpha_i   \bden(\hat{S}_i,\hat{S}_i) \right)  - \frac{2d(r+1) \rho}{n}\\
&\geq \left( \sum_{i=1}^r  \alpha_i   \bden(\hat{S}_i,\hat{S}_i) \right) - \frac{2d(r+1) \rho}{n}\\
&= \left( \sum_{i=1}^r  \alpha_i   \left(1-\frac{1}{|\hat{S}_i|}\right)  \right) - \frac{2d(r+1) \rho}{n}\\
&=  \left( \sum_{i=1}^r  \alpha_i  \right)   - \frac{2d(r+1) \rho + r}{n}\\
&\geq  0.9  -  \frac{5dr \rho}{n}\\
&=  0.8\\
\end{align*}


\subsection{Proof of Lemma \ref{lem:zerorecover}}

We prove the Lemma by exhibiting an algorithm  which, given as input a graph $G \sim \G(n,\frac{1}{2},k,r)$ and a signaling scheme $\sig$ attaining the claimed attacker utility, outputs a family of clusters $\T$ from which a constant fraction of the planted cliques can be recovered in polynomial time. Specifically, we show that for a constant fraction of the planted cliques $S_1,\ldots,S_r$, there is a cluster $T \in \T$ which satisfies the requirements of Lemma \ref{lem:approx_recovery}.  
To simplify our proof, we assume that $\rho d/2$ is an integer, though unsurprisingly this assumption can easily be removed.
\begin{algorithm}[h]
\caption{Algorithm for Computing Approximation of Planted Cliques from a Signaling Scheme}
\label{alg:zerorecover1}
\begin{algorithmic}[1]
\INPUT Graph $G=(V,E)$ with adjacency matrix $A \in \RR^{n \times n}$
\INPUT A signaling scheme $\sig: V \to \Sigma$, represented in convex decomposition form $(\alpha,\dii)$ where $\alpha \in \Delta_\Sigma$, and $\dii_\sigma \in \Delta_V$ for each $\sigma \in \Sigma$.
\OUTPUT A list $\T$ of subsets of $V$, each of size $\rho d / 2$.
\STATE For each $\sigma \in \Sigma$, compute the attacker's minimax strategy $y_\sigma \in \Delta_V$ in the subgame corresponding to signal $\sigma$.
\STATE For each $\sigma \in \Sigma$, discard from each of $x_\sigma$ and $y_\sigma$ all entries greater than $\frac{2}{\rho d}$. Define
\[\hat{x}_\sigma(v) = 
\begin{cases}
   x_\sigma(v)  \mbox{  if } x_\sigma(v) \leq \frac{2}{\rho d}, \\
   0 \mbox{  otherwise. } 
\end{cases}
\]
\[\hat{y}_\sigma(v) = 
\begin{cases}
   y_\sigma(v)  \mbox{  if } y_\sigma(v) \leq \frac{2}{\rho d}, \\
   0 \mbox{  otherwise. } 
\end{cases}
\]
\STATE For each $\sigma \in \Sigma$, let $z_\sigma \in \Delta_V$ be an extreme-point solution of the following linear program. (Note that $\hat{x}_\sigma$ is fixed)
\begin{lp}\label{lp:recover1}
  \maxi{(\hat{x}_\sigma)^\intercal A z_\sigma}
\st
\con{||z_\sigma||_\infty \leq \frac{2}{\rho d}}
\con{\sum_{v \in V} z_\sigma(v)\leq 1}
\qcon{z_\sigma(v) \geq 0}{v \in V}
\end{lp}
\STATE Let $T_\sigma$ be the support of $z_\sigma \in \Delta_V$, for each $\sigma \in \Sigma$.
\STATE Output $\T= \set{ T_\sigma : \sigma \in \Sigma}$.
 \end{algorithmic}
\end{algorithm}

Consider Algorithm \ref{alg:zerorecover1}. In step (2), the algorithm ``zeroes-out'' those vertices which are  ``overrepresented'' in either the attacker's strategy or the posterior distribution over states of nature.  This is justified because, as will become clear in the proof, we designed the game's payoffs so that overrepresented nodes can easily be protected by the defender, and hence account for only a small fraction of the attacker's total utility. After zeroing-out large entries, the posteriors $\hat{x}_\sigma$ and attacker strategies $\hat{y}_\sigma$ have entries no larger than $1/\Theta(k)$, and therefore behave roughly as uniform distributions over $\Theta(k)$ vertices.  Steps (3) and (4) are  ``cleanup steps'', which convert the attacker's mixed strategy to a uniform distribution over $\frac{\rho d}{2}  = \Theta(k)$ vertices $T_\sigma$. We will show that the resulting family $\T$ of clusters must overlap substantially with the planted cliques, in the sense of Lemma \ref{lem:approx_recovery}, for the utility of the attacker  to exceed the background density of $\frac{1}{2}$ by a constant. We formalize all this through a sequence of propositions.

\begin{prop}\label{prop:zero1}
  The attacker's utility $u(\sig)$ from signaling scheme $\sig=(\alpha,x)$ satisfies \[ u(\sig) \leq \sum_{\sigma \in \Sigma} \alpha_\sigma \hat{x}_\sigma^\intercal A \hat{y}_\sigma\]
where $A$ is the adjacency matrix of $G$, and $\hat{x}_\sigma$ and $\hat{y}_\sigma$ are as defined in step (2) of Algorithm \ref{alg:zerorecover1}.
\end{prop}
\begin{proof}
We prove this bound signal-by-signal. Fix a signal $\sigma \in \Sigma$ with posterior $x_\sigma$ and attacker strategy $y_\sigma$, and as shorthand use $x=x_\sigma$, $y=y_\sigma$, $\hat{x}= \hat{x}_\sigma$, and $\hat{y}=\hat{y}_\sigma$. We distinguish nodes which are \emph{overrepresented} in $x$ or $y$ --- specifically, the nodes $O=\set{v: \max(x(v), y(v)) \geq \frac{2}{\rho d}}$. There are at most $\rho d$ such nodes, since each of $x$ and $y$, by virtue of being a probability distribution,  can have at most $\rho d / 2$ entries exceeding $2/ \rho d$. 

We consider a defender who  chooses $D \sse O$ with $|D| = \min(d, |O|)$ uniformly at random. Observe that such a defender protects each node in $O$ with probability at least $\frac{\min(d,|O|)}{|O|} \geq \frac{d}{\rho d} = \frac{1}{\rho}$. The utility of the attacker is equal to the probability that $\theta \sim x$ and $a \sim y$ share an edge, minus $\rho$ times the quantity $\Pr [a \in D] + \Pr[\theta \in D]$, which can be bounded as follows.
\begin{align}\label{eq:zero1}
  u^\sigma(\sig) &= x^\intercal A y - \rho (\Pr [ \theta \in D ] + \Pr[ a \in  D]) \notag \\ 
  &= x^\intercal A y - \rho (\Pr [ \theta \in O] \Pr [ \theta \in D | \theta \in O] + \Pr [ a \in O] \Pr [a \in D | a \in O]) \notag \\
  &\leq x^\intercal A y - \rho . \frac{1}{\rho} (\Pr [ \theta \in O] + \Pr [ a \in O] ) \notag \\
  &= x^\intercal A y -  (x(O) + y(O)) 
\end{align}
Next, we let $\bar{x} = x - \hat{x}$ and $\bar{y} = y - \hat{y}$. Note that $x(O) = ||\bar{x}||_1$ and $y(O) = ||\bar{y}||_1$. We can bound $x^\intercal A y$ as follows.
\begin{align}\label{eq:zero2}
 x^\intercal A y &= (\hat{x} + \bar{x})^\intercal A (\hat{y} + \bar{y}) \notag\\
 &= \hat{x}^\intercal  A \hat{y} +  \hat{x}^\intercal  A \bar{y} +  \bar{x}^\intercal  A \hat{y} +  \bar{x}^\intercal  A \bar{y}  \notag\\
 &\leq \hat{x}^\intercal  A \hat{y} +  \hat{x}^\intercal  \bar{y} +  \bar{x}^\intercal  \hat{y} +  \bar{x}^\intercal  \bar{y}  \notag\\
 &\leq \hat{x}^\intercal  A \hat{y} +  \hat{x}^\intercal  \bar{y} +  \bar{x}^\intercal  \hat{y} +  2\bar{x}^\intercal  \bar{y}  \notag\\
 &= \hat{x}^\intercal  A \hat{y} +  \bar{x}^\intercal y + \bar{y}^\intercal x   \notag\\
 &\leq \hat{x}^\intercal  A \hat{y} +  ||\bar{x}||_1 ||y||_\infty + ||\bar{y}||_1 ||x||_\infty  \mbox{\hspace{1cm} ( Holder's inequality) }  \notag\\
 &\leq \hat{x}^\intercal  A \hat{y} +  ||\bar{x}||_1   + ||\bar{y}||_1    \notag\\
 &= \hat{x}^\intercal  A \hat{y} +  x(O) + y(O)   
\end{align}
Combining \eqref{eq:zero1} with \eqref{eq:zero2} completes the proof.
\end{proof}

\begin{prop}\label{prop:zero2}
\[ u(\sig) \leq \sum_{\sigma \in \Sigma} \alpha_\sigma \hat{x}_\sigma^\intercal A z_\sigma\]
\end{prop}
\begin{proof}
  This follows from Proposition \ref{prop:zero1} and the fact that that $\hat{y}_\sigma$ is a feasible solution to  linear program~\eqref{lp:recover1}.
\end{proof}

\begin{prop}\label{prop:zero3}
  $z_\sigma$ is a uniform distribution over $\rho d / 2$ vertices, namely $T_\sigma$.
\end{prop}
\begin{proof}
  Recall that we assumed that $\rho d / 2$ is an integer. Therefore, a simple argument shows that an extreme-point solution to LP \eqref{lp:recover1} must set each variable either to $0$ or to $\frac{2}{\rho d}$. In particular, the optimal solution must set precisely $\rho d / 2$ entries of $z_\sigma$ to $\frac{2}{\rho d}$, as needed.
\end{proof}

We now upperbound  the contribution of edges outside the planted cliques to the utility of the attacker. Write the adjacency matrix $A$ as $A=A^- + A^+$, where $A^-$ are the background edges added in Step (1) of Definition \ref{def:plantedcover}, and $A^+$ are the clique edges added in Step (2) of Definition \ref{def:plantedcover}.

\begin{prop}\label{prop:zero4}
The following holds with high probability for all signals $\sigma \in \Sigma$.
  \[ (\hat{x}_\sigma)^\intercal A^- z_\sigma \leq 0.5 + \frac{\epsilon}{2}\]
\end{prop}
\begin{proof}
  Pick an arbitrary signal $\sigma$. Since $||\hat{x}_\sigma||_\infty \leq \frac{2}{\rho d}$, we have that $(\hat{x}_\sigma)^\intercal A^- z_\sigma$ is bounded from above by the value of the following linear program, in which $z_\sigma$ is held fixed and $x \in \RR^n$ is allowed to vary.
\begin{lp}
  \maxi{x^\intercal A z_\sigma}
\st
\con{||x||_\infty \leq \frac{2}{\rho d}}
\con{\sum_{v\in V} x(v) \leq 1}
\qcon{ x(v) \geq 0}{v \in V}
\end{lp}
By an argument similar to in Proposition \ref{prop:zero3}, at optimality $x$ is a uniform distribution over some $\rho d /2$ vertices $R$. Therefore, since $z_\sigma$ is also a uniform distribution over the $\rho d / 2$ vertices $T_\sigma$ (Proposition \ref{prop:zero3}), the value of the above linear program is equal to $\bden(R,T_\sigma)$. Since $\rho d / 2=\Theta(k) = \omega(\log n)$,  Proposition \ref{prop:bidensity-whp} implies the claimed bound.
\end{proof}

We can now wrap up the proof of Lemma \ref{lem:zerorecover} using the above propositions. We will show that, on average over our planted cliques $S_1,\ldots,S_r$, there is some $T \in \T$ output by Algorithm \ref{alg:zerorecover1} overlapping with a constant fraction of the clique vertices. As notation, we let $A^i$ denote the adjacency matrix of the clique $S_i$, for $i = 1, \ldots,r$. Note that $A^+ \leq \sum_{i=1}^r A^i$. First, we show that foreground edges contribute $\Omega(\epsilon)$ to the outerproduct $(\hat{x}_\sigma)^\intercal A z_\sigma$, on average over signals $\sigma \in \Sigma$.
\begin{align*}
0.5 + \eps &\leq u(\sig) \\
&\leq \sum_{\sigma \in \Sigma} \alpha_\sigma \hat{x}_\sigma^\intercal A z_\sigma &\mbox{(Proposition \ref{prop:zero2})}\\
&= \sum_{\sigma \in \Sigma} \alpha_\sigma \hat{x}_\sigma^\intercal (A^- + A^+) z_\sigma\\
&\leq (0.5 + \frac{\eps}{2}) \sum_{\sigma \in \Sigma} \alpha_\sigma  + \sum_{\sigma \in \Sigma} \alpha_\sigma \hat{x}_\sigma^\intercal A^+ z_\sigma &\mbox{(Proposition \ref{prop:zero4})}\\
&= 0.5 + \frac{\eps}{2}+ \sum_{\sigma \in \Sigma} \alpha_\sigma \hat{x}_\sigma^\intercal A^+ z_\sigma
\end{align*}
\noindent Therefore $\sum_{\sigma \in \Sigma} \alpha_\sigma \hat{x}_\sigma^\intercal A^+ z_\sigma \geq \frac{\eps}{2}$. Next, we ``break up'' this sum of outerproducts into the constituent contributions of each planted clique.
\begin{align*}
\frac{\eps}{2} &\leq  \sum_{\sigma \in \Sigma} \alpha_\sigma \hat{x}_\sigma^\intercal A^+ z_\sigma \\
&\leq  \sum_{\sigma \in \Sigma} \alpha_\sigma \hat{x}_\sigma^\intercal (\sum_{i=1}^r A^i) z_\sigma \\
&= \sum_{i=1}^r  \sum_{\sigma \in \Sigma} \alpha_\sigma \hat{x}_\sigma^\intercal A^i z_\sigma \\
&\leq \sum_{i=1}^r  \sum_{\sigma \in \Sigma} \alpha_\sigma \hat{x}_\sigma (S_i) z_\sigma(S_i) \\
&= \sum_{i=1}^r  \sum_{\sigma \in \Sigma} \alpha_\sigma \hat{x}_\sigma (S_i) \frac{|T_\sigma \intersect S_i|}{|T_\sigma|} 
\end{align*}
\noindent Finally, we show that the average planted clique is well represented by some $T \in \T$.
\begin{align*}
\frac{\eps}{2} &\leq \sum_{i=1}^r  \sum_{\sigma \in \Sigma} \alpha_\sigma \hat{x}_\sigma (S_i) \frac{|T_\sigma \intersect S_i|}{|T_\sigma|} \\
 &\leq  \sum_{i=1}^r  \left(\sum_{\sigma \in \Sigma} \alpha_\sigma \hat{x}_\sigma (S_i)\right) \left(\max_{T \in \T}  \frac{|T \intersect S_i|}{|T_\sigma|} \right) \\
 &\leq  \sum_{i=1}^r  \frac{|S_i|}{n} \left(\max_{T \in \T}  \frac{|T \intersect S_i|}{|T_\sigma|} \right) &\mbox{(Since $\sum_\sigma \alpha_\sigma \hat{x}_\sigma (v) \leq \sum_\sigma \alpha_\sigma x_\sigma(v) = \di_v = \frac{1}{n}$) } \\
 &= \frac{k}{n} \sum_{i=1}^r  \left(\max_{T \in \T}  \frac{|T \intersect S_i|}{|T_\sigma|} \right) \\
 &= O(1) \cdot \frac{1}{r}  \sum_{i=1}^r  \left(\max_{T \in \T}  \frac{|T \intersect S_i|}{|T_\sigma|} \right) &\mbox{(Since $r = \Theta(\frac{n}{k})$)} \\
 &= O(1) \cdot  \frac{1}{r}  \sum_{i=1}^r  \left(\max_{T \in \T}  \frac{|T \intersect S_i|}{k} \right) &\mbox{(Since $|T_\sigma| = \rho d / 2 = \Theta(k)$)} \\
\end{align*}
Therefore $\frac{1}{r}  \sum_{i=1}^r  \left(\max_{T \in \T}  \frac{|T \intersect S_i|}{k} \right) \geq \Omega(\eps)$; i.e., the average planted clique intersects at least one of the sets in $\T$ in a constant fraction $\Omega(\epsilon)$ of its vertices. This implies that an $\Omega(\epsilon)$ fraction of the planted cliques  intersect at least one of the sets in $\T$ in an $\Omega(\epsilon)$ fraction of its vertices, by a simple counting argument. Since each set in $\T$ is of size $\Theta(k)$, this satisfies the requirements of Lemma \ref{lem:approx_recovery} for a constant fraction of the planted cliques $S_1,\ldots,S_r$, completing the proof.

\section{Conclusion}

Our results raise several open questions, and we close with some of them.
First, note that our impossibility result for explicit zero sum games rules out approximation algorithms with additive error on the order of $\frac{1}{\polylog n}$ relative to the range of player payoffs, where $n$ is the number of a player strategies. Essentially, this rules out a \emph{fully polynomial time approximation scheme} (FPTAS) for optimal signaling, but not a PTAS: an $\epsilon$ additive approximation algorithm for every constant $\epsilon$ independent of $n$. A more immediately attainable goal might be a quasipolynomial time approximation scheme (QPTAS), since a quasipolynomial time algorithm exists for the planted clique problem. In the event of such positive results, it makes sense to examine extensions to more general (multi-player and non-zero sum) games.

 The reader might have noticed some similarities between optimal signaling and the problem of optimizing over Nash equilibria, which is known to admit a QPTAS, but not a PTAS by  the reduction of \citet{hazankrauthgamer} from the planted clique problem. However, those similarities do not appear to give any immediate answers for the signaling question. For example, whereas adaptations of the  QPTAS for the best Nash equilibrium problem might be able to identify individual ``good'' signals, piecing those together into an information structure appears to be nontrivial.

Our results also point towards other potential signaling questions. Much of the related work mentioned in the introduction considers particular game domains in which signaling can be used to effect desirable outcomes. We believe that our ideas, and in particular the connection with dense subgraph detection, can shed light on the algorithmic component of some of these applications --- we single out the auction domains studied in \cite{emeksignaling,miltersensignals,DIR14} as likely targets.



\section*{Acknowledgments}
 We thank Aditya Bhaskara, Uriel Feige, Nicole Immorlica, Albert Jiang, and Shang-Hua Teng for helpful discussions. We also thank the anonymous FOCS reviewers for helpful feedback and suggestions.

{
\bibliography{agt}
\bibliographystyle{plainnat}     
}

\newpage
\appendix
\section{Proof of Lemma \ref{lem:approx_recovery} }
\label{app:approx_recovery}

Let $G=(V,E)$. We fix one of the planted $k$-cliques $S \sse V$, and prove Lemma \ref{lem:approx_recovery} for every $T$ satisfying $|T \intersect S| > \epsilon |T \union S|$. Our proof has two steps. First, we show that if we can sample uniformly from $T \intersect S$, then we can recover planted clique $S$ in polynomial time with high probability. Then we show that we can simulate sampling from $T \intersect S$ by combining sampling from $T$ with some brute force enumeration in polynomial time.

\subsection{Step 1}

Assume we can sample uniformly from $T \intersect S$. Our algorithm for recovering $S$ is as follows.

\begin{itemize}
\item Sample: Let $R$ be a sample of $200\log n$ vertices from $T \intersect S$.
\item Filter 1: Let $\tilde{S}$ be all the common neighbors of $R$.
\item Filter 2: Let $\hat{S}$ be those nodes in $\tilde{S}$ with at least $k-1$ neighbors in $\tilde{S}$.
\end{itemize}

It is clear that nodes in the $k$-clique $S$ survive both filtration steps, and therefore $\hat{S} \supseteq S$. We will show that in fact $\hat{S} = S$, with high probability. For this, we need the following propositions bounding the ``connectivity'' of nodes in $V \sm S$ into $S$ and $T \intersect S$. We partition the edges of $G$ into \emph{background edges} $E^-$, added in step (1) of Definition \ref{def:plantedcover}, and \emph{foreground edges} $E^+$ added in step~(2). 

\begin{prop}\label{prop:filtera}
For every $T$, the number of vertices $v$ outside $S$ with $|E^-(v,S \intersect T)| > 0.6 |T \intersect S|$ is $O(\log n)$, with high probability. 
\end{prop}
\begin{proof}
Since $p \leq 0.5$ and $|T \intersect S| = \Omega(k) = \omega(\log n)$, this follows directly from Proposition \ref{prop:bidensity-whp}.
\end{proof}

\begin{prop}\label{prop:filterb}
There are no vertices $v$ outside $S$ with $|E^-(v,S)| > 0.6 k$, with high probability.
\end{prop}
\begin{proof}
This is by a straightforward application of the Chernoff bound and the union bound.
\end{proof}

\begin{prop}\label{prop:filterc}
Every vertex $v$ outside $S$ has $|E^+(v,S)| = O(\log^2 n) = o(k)$, with high probability.
\end{prop}
\begin{proof}
This follows from two facts which hold with high probability: (a) no two planted cliques overlap in more than $O(\log n)$ vertices, and (b) no vertex is in more than $O(\log n)$ planted cliques. We show both next.

For (a),  let $S_1,\ldots,S_r$ be the planted cliques (one of which is our $S$).  Fix $S_i$ and $S_j$ with $i \neq j$. The probability a vertex $v$ is in both $S_i$ and $S_j$ is $(\frac{k}{n})^2 < \frac{1}{n}$.  A simple application of the Chernoff bound implies that $|S_i \intersect S_j|  < 4 \log n$ with probability at least $1-1/n^4$.\footnote{This also requires showing negative association of the relevant indicator variables --- one per node of the graph, indicating membership in both $S_i$ and $S_j$. We omit the (straightforward) details.} By the union bound and the fact that $r< n$, this holds for all pairs $i \neq j$ with probability at least $1-1/n^2$.

For (b), observe that for a fixed vertex $v$ the events $\set{v \in S_i}_{i=1}^r$  are independent Bernoulli trials with probability $\frac{k}{n}$ each. Since the number of these events is $r= O(\frac{n}{k})$, the Chernoff bound implies that at most $O(\log n)$ of these events hold with high probability $1-1/\poly(n)$. Taking a union bound over all vertices $v$ completes the proof.

\end{proof}

We can now show that no vertex $v$ outside $S$ survives both filtration steps, with high probability. First, combining Propositions \ref{prop:filtera} and \ref{prop:filterc}, we know that for all but at most $O(\log n)$ nodes $v \not\in S$ we have that $|E(v,T \intersect S)| =E^-(v,T \intersect S) + E^+(v,T \intersect S) \leq 0.6 |T \intersect S| + o(k) \leq 0.7 |T \intersect S|$ (for sufficiently large $k$). Since $R$ is a random sample of $200\log n$ vertices from $|T \intersect S|$, a simple application of the Chernoff bound shows that $E(v,R) \leq 1.1 \frac{|R|}{|T \intersect S|} E(v,T \intersect S)$ for every vertex $v$, with high probability. Therefore, for those vertices with $|E(v,T \intersect S)| \leq 0.7 |T \intersect S|$, none of them connect to more than $0.8 |R|$ vertices in our sample $R$, with high probability. Therefore, at most $O(\log n)$ vertices outside $S$ survive the first filtration step --- i.e., $|\tilde{S} \sm S| = O(\log n)$ --- with high probability.

Second, since $|\tilde{S} \sm S| = O(\log n)$, every node $v \in \tilde{S} \sm S$ has $|E(v, \tilde{S})| \leq |E(v,S)| + O(\log n) = |E(v,S)| + o(k)$. Proposition \ref{prop:filterb} implies that $|E(v,S)| \leq 0.6 k$ with high probability, and therefore no vertex $v \in \tilde{S} \sm S$ survives the second filtration step (for sufficiently large $k$).

\subsection{Step 2}

In step 1, we assumed the ability to sample $200 \log n$ vertices uniformly at random from $T \intersect S$. However, since $S$ is a-priori unknown, this is impossible in the most naive sense. Nevertheless, such sampling can be simulated efficiently as follows: Sample roughly $\frac{200}{\epsilon} \log n$ vertices uniformly from $T$, and attempt the algorithm of step $1$ on every subset of roughly $ 200 \log n$ of the sampled vertices. This runs in polynomial time. Moreover, since $|T \intersect S| \geq \epsilon |T|$, with high probability it is the case that roughly $ 200 \log n$ of the sampled vertices lie in $T \intersect S$, and are distributed uniformly therein. 


\section{Omitted Preliminaries}
\label{app:prelim}

\subsection*{Probabilistic Bounds on Random Graphs}

\begin{prop}\label{prop:pairdensity}
    Let $G \sim \G(n,p)$, and let $A$ be the adjacency matrix of $G$. For any family $P \sse [n] \cross [n]$ of pairs of vertices, and $\alpha\geq 1$,
\[ \Pr\left[ \frac{1}{|P|} \sum_{(i,j) \in P} A_{ij}  \geq  \alpha p\right] \leq 2 \exp\left(-\frac{(\alpha-1)^2 p |P|}{12\alpha}\right)\]
\end{prop}
\begin{proof}
  For an unordered pair $e \in \binom{[n]}{2}$, let $x_e$ denote the indicator variable for the event $e \in E(G)$. By definition of $\G(n,p)$, the indicator variables $x_e$ are independent Bernoulli random variable with parameter $p$. Observe that 
  \begin{align}\label{eq:pairspf1}
    \sum_{(i,j) \in P} A_{ij} &= \sum_{(i,j) \in X \cross Y} x_{\set{i,j}} 
  \end{align}
Since our edges are undirected, whereas the above sum is over ordered pairs, the number of times any of our indicator variables $x_e$ appears in the above sum varies. Specifically, for $e=\set{i,j}$, the variable $x_e$ appears twice if both $(i,j)$ and $(j,i)$ are in $P$, once if only one of $(i,j)$ and $(j,i)$ is in $P$, and zero times otherwise. Let $F_1$ be the set of edges appearing at least once, and $F_2$ be the set of edges appearing twice. 
Observe that $|F_2| \leq |F_1| \leq |P|$ and 
$|F_1| + |F_2| \leq |P|$. 
We can now simplify expression \eqref{eq:pairspf1}.
\begin{align}\label{eq:pairspf2}
    \sum_{(i,j) \in P} A_{ij} &= \sum_{e \in F_1} x_e + \sum_{e \in F_2} x_e 
  \end{align}
Observe that each of $m_1=\sum_{e \in F_1} x_e$ and $m_2=\sum_{e \in F_2} x_e$ is the sum of independent Bernoulli random variables with parameter $p$.
Using the Chernoff bound (Proposition \ref{prop:chernoff}, Equation \eqref{eq:chernoffadd}), we can bound the probability that  $m_1$ is too large. 
\begin{align*}
\Pr [ m_1 \geq p |F_1| + 0.5 (\alpha-1) p |P|] &= \Pr [ m_1 \geq \Ex[m_1] + 0.5(\alpha-1) p |P|] \\
&\leq \exp\left(-\frac{0.25(\alpha-1)^2 p^2 |P|^2}{ (2 p |F_1| + 0.5 (\alpha-1) p |P|) }\right)\\
&= \exp\left(-\frac{0.25(\alpha-1)^2 p |P|^2}{2 |F_1| + 0.5 (\alpha-1)  |P| }\right) \\
&\leq \exp\left(-\frac{0.25(\alpha-1)^2 p |P|^2}{2 |F_1| + \alpha  |P| }\right) \\
&\leq \exp\left(-\frac{0.25(\alpha-1)^2 p |P|^2}{2 \alpha |F_1| + \alpha  |P| }\right) \\
&\leq \exp\left(-\frac{0.25(\alpha-1)^2 p |P|^2}{3 \alpha |P| }\right) \\
&= \exp\left(-\frac{(\alpha-1)^2 p |P|}{12 \alpha }\right) 
\end{align*}
An essentially identical calculation yields a similar bound for $m_2$.
\begin{align*}
\Pr [ m_2 \geq p |F_2| + \frac{1}{2}(\alpha-1) p |P|]  &\leq \exp\left(-\frac{(\alpha-1)^2 p |P|}{12\alpha}\right)  
\end{align*}
Using the union bound, we complete the proof.
\begin{align*}
\Pr\left[ \frac{1}{|P|} \sum_{(i,j) \in P} A_{ij}  >  \alpha p\right] &= \Pr [ m_1 + m_2 > \alpha p |P|] \\
&\leq \Pr [ m_1 + m_2 > p (|F_1| + |F_2|) + (\alpha - 1) p |P|]\\
&\leq \Pr [ m_1 > p |F_1| + \frac{1}{2}(\alpha - 1) p |P|] +  \Pr [ m_2 > p |F_2| + \frac{1}{2}(\alpha - 1) p |P|]\\
&\leq 2 \exp\left(-\frac{(\alpha-1)^2 p |P|}{12\alpha}\right)
\end{align*}
\end{proof}

\begin{prop}[Proposition \ref{prop:bidensity-whp} in main body]
    Let $p \in (0,1)$ and $\alpha > 1 $ be absolute constants (independent of $n$), and let $G \sim \G(n,p)$. There is an absolute constant $\beta = \beta(p,\alpha)$ such that the following holds with high probability for all clusters $X$ and $Y$ with $|X|,|Y| > \beta \log n$. 
\[ \bden_G(X,Y) \leq \alpha p \]
\end{prop}
\begin{proof}
  Consider two cluster sizes $k, \ell \in [n]$. Given a fixed pair of clusters $|X|$ and $|Y|$ with $|X| = k$ an $|Y|=\ell$, Proposition \ref{prop:pairdensity} implies that
  \begin{align*}
    \Pr\left[ \bden(X,Y) \geq \alpha p \right] &= \Pr\left[ \frac{1}{|X| |Y|} \sum_{(i,j) \in X \cross Y} A_{ij}  \geq  \alpha p\right] \\
      &\leq 2 \exp\left(-\frac{(\alpha-1)^2 p k \ell}{12\alpha}\right)    \\
  \end{align*}
In particular, since $p$ and $\alpha$ are constants, there is a constant $c>0$ such that \[\Pr\left[ \bden(X,Y) \geq \alpha p\right] \leq \exp(- c k \ell).\]  

There are at most $n^k$ clusters of size $k$, and $n^\ell$ clusters of size $\ell$. Therefore, the probability that any pair of clusters $X$ and $Y$ with $|X| = k$ and $|Y|=\ell$ satisfies $\bden(X,Y) \geq \alpha p$ is, by the union bound, at most \[ n^{k+\ell} \exp( - c k \ell) = \exp((k + \ell) \log n - c k \ell).\]
When $k,\ell \geq \beta \log n$ for some constant $\beta$, this probability is at most \[ \exp((k + \ell) \log n - c k \ell) \leq \exp (2 \beta \log^2n - \beta^2 c \log^2n).\]
We can choose the constant $\beta$ large enough so that this probability is $\exp(-\Omega(\log^2 n))$ for every pair of integers $k,\ell \geq \beta \log n$. Taking a union bound over all pairs of integers $k,\ell \in [n]$ with  $k,\ell \geq \beta \log n$, of which there are at most $n^2$, completes the proof.
\end{proof}

\subsection*{Tail Bounds}
 We use the following convenient forms of the Chernoff bound
\begin{proposition}[Chernoff Bounds]\label{prop:chernoff}
  Let $X_1,\ldots, X_n$ be independent Bernoulli random variables with $\mu= \sum_{i=1}^n  \Ex [X_i]$.  Let $X = \sum_{i=1}^n  X_i$ be their sample sum.  For every $\alpha \geq 1$  the following holds.
  \begin{equation}
    \label{eq:chernoffmult}
    \Pr [ X > \alpha \mu ]\leq \exp\left(-\frac{(\alpha -1)^2 \mu }{\alpha + 1}\right)
  \end{equation}
Equivalently, for every $t \geq 0$, the following holds.
\begin{equation}
  \label{eq:chernoffadd}
 \Pr [ X >  \mu + t ]\leq \exp\left(- \frac{t^2}{2\mu + t} \right) 
\end{equation}
\end{proposition}

\end{document}